\documentclass[12pt]{article}

\usepackage{graphics}
\usepackage{epsfig,amsmath,amsthm,amsfonts,amssymb}
\usepackage{latexsym,url}
\usepackage{capt-of}

\newtheorem{theorem}{Theorem}
\newtheorem{proposition}[theorem]{Proposition}

\begin{document}

\title{Permanence and Stability of a Kill the Winner Model in Marine Ecology}
\author{Daniel A. Korytowski and Hal L. Smith}
%
%
\maketitle

\abstract{We focus on the long term dynamics of ``killing the winner'' Lotka-Volterra  models of marine communities consisting of bacteria, virus,  and zooplankton. Under suitable conditions, it is shown that there is a unique equilibrium with all populations present which is stable, the system is permanent, and the limiting behavior of its solutions is strongly constrained.
}

keywords: Virus, bacteria, zooplankton, kill the winner, infection network, Lotka-Volterra system, permanence.


\section{Introduction}
\label{sec:1}
It is now known that the microbial and viral communities in marine environments are remarkably diverse
but are supported by relatively few nutrients in very limited concentrations \cite{Su,W}. What can explain the observed diversity?
What prevents the most competitive bacterial strains from achieving large densities at the expense of less competitive strains?
Thingstad \cite{Th,Th2,Th3,Th4} has suggested that virus impose top down control of bacterial densities. Together with various
coauthors, he has described an idealized food web consisting of bacteria, virus and zooplankton to illustrate mechanisms
of population control, referred to as ``killing the winner'' since any proliferation of a ``winning'' bacterial strain results in
increased predation by some virus. The kill the winner (KtW) mathematical model of this scenario, in the form of a system of Lotka-Volterra
equations for bacterial, virus, and zooplankton densities is, as noted by Weitz \cite{W}, based on the assumptions that
(1) all microbes compete for a common resource, (2) all microbes, except for one population, are susceptible to virus infection,
(3) all microbes are subjected to zooplankton grazing, (4) viruses infect only  a single type of bacteria.

Various forms of the KtW model have appeared in the work of Thingstad et al. \cite{Th,Th2,Th3,Th4} and recently in
the monograph of Weitz \cite{W}. As the nutrient level can be assumed to be in quasi-steady state with consumer densities, the models
typically involve only the $n$ bacteria types,  $n-1$ virus types, and a single zooplankton. While the literature contains many numerical simulations of KtW solutions, very little is know about the long term behavior of these solutions. It is the aim of this paper to initiate a mathematical
analysis of this important model system. We will show that the equilibrium with all populations present is unique and stable to small perturbations, that the system is permanent
in the sense that all population densities are ultimately bounded away from extinction by an initial condition independent positive quantity, and that
the long-term average of each population's density is precisely equal to its corresponding positive equilibrium value. In addition, we are able to
provide some qualitative information about the long term dynamics. It is shown that the zooplankton density and the density of the bacterial strain
resistant to virus infection converge to their equilibrium value. Furthermore, if a solution does not converge to the positive equilibrium, then its long-term dynamics can be described by an uncoupled system consisting of $n-1$ conservative two-species systems involving each virus-susceptible bacteria and its associated virus. This implies that non-convergent solutions are, at worst, quasi-periodic.

Thingstad notes in \cite{Th2} that a weakness of the killing the winner hypothesis is the assumption (4) that each virus infects only a single type of bacteria. Indeed, recent data \cite{Fl,Jover,W} suggests that some virus have large host range. We will also show that most of our conclusions stated above
hold without the restriction (4). For example, they hold for a nested infection network.

The results described above allow one to determine a plausible route by which a KtW community (satisfying $(1)-(4)$) consisting of $n$ bacterial
strains, $n-1$ virus strains, and a single non-specific zooplankton grazer might be assembled
starting with a community consisting of a single bacteria and its associated virus and subsequently adding one new population at a time until the final community is achieved. By a plausible route, we require that each intermediate community be permanent \cite{HS}, also called uniformly persistent \cite{SZ,T,ST}, since
a significant time period may  be required to make the transition from one community to the next in the succession and therefore each community must be resistant to extinctions of its members. In \cite{KS,KS2}, considering only bacteria and virus communities, we established a plausible route to the assembly of a community consisting of $n$ bacterial
strains and either $n$ or $n-1$ virus strains in which the infection network is one to one under suitable conditions. See also \cite{HaMiSn} although they did not infer permanence.  Therefore, since we merely need to add zooplankton to  community consisting of $n$ bacteria and $n-1$ virus, the main result of this paper ensures that there is a plausible assembly path to the KtW community.

In the next section, we formulate our KtW model and state our main results. Technical details are include in a final section.

\section{The KtW model}

Our KtW model, consisting of $n\ge 2$ bacterial types, $n-1$ virus types and one zooplankton, is patterned after equations (7.28) in \cite{W} with slight changes. Densities of bacteria strains are denoted by $B_i$, virus strains by $V_i$,
and zooplankton by $Y$. The difference in our model and  (7.28) is in the way
that inter and intra-specific competition among bacteria
is modeled. We assume that the
density dependent reduction in growth rate due to competition  is identical for all bacterial strains as in \cite{KS,KS2b,KS2}. Virus strain $V_i$ infects bacterial strain $B_i$ for $i\ne n$ but $B_n$ is resistent to virus infection.
Zooplankton graze on bacteria at a strain independent rate. Virus adsorption rate is $\phi_i$ and burst size is $\beta_i$; $w$ represents a common loss rate. The equations follow.

\begin{eqnarray}\label{LVmono}
B_i' &=& B_i(r_i-w-aB )-B_i\phi_iV_i-\alpha B_i Y,  \nonumber\\
V_i' &=& V_i(\beta_i \phi_i B_i - k_i-w),\ 1\le i \le n-1 \nonumber\\
B_n' &=& B_n(r_n-w-aB)-\alpha B_n Y\\
Y'   &=& Y(\alpha \rho B - w - m),  \nonumber
\end{eqnarray}
where $B=\sum_j B_j$ is the sum of all bacterial densities.

It is convenient to scale variables as:
$$
P_i=\phi_iV_i, \ H_i=aB_i,\ Z=\alpha Y,
$$
and parameters as
$$
n_i=\beta_i\phi_i/a,\ \lambda=\alpha \rho/a,\ e_i=\frac{k_i+w}{n_i}, \ q=\frac{w+m}{\lambda}.
$$
This results in the following scaled system where $H=\sum_j H_j$:

\begin{eqnarray}\label{eqns-scaled}
H_i' &=& H_i(r_i-w-H)-H_iP_i-H_i Z, \nonumber\\
P_i' &=& n_iP_i(H_i-e_i),\ 1\le i \le n-1 \nonumber\\
H_n' &=& H_n(r_n-w-H)-H_n Z\\
Z'   &=& \lambda Z(H - q).  \nonumber
\end{eqnarray}
Only positive solutions of \eqref{eqns-scaled} with $H_i(0)>0, P_j(0)>0, Z(0)>0$ for all $i,j$ are of interest. It is then evident that
$H_i(t)>0, P_j(t)>0, Z(t)>0$ for all $t$ and $i,j$.

There is a unique  positive equilibrium $E^*$ if and only if the virus-resistent microbe $H_n$ has the lowest growth rate among the bacteria
\begin{equation}\label{wkBn}
w<r_n<r_j,\  1 \le j \le n-1,
\end{equation}
and if
\begin{equation}\label{hiresource}
\sum_{i=1}^{n-1} e_i<q<r_n-w.
\end{equation}
Then $E^*$ is given by
\begin{equation}
H_j^*= e_j, P_j^*= r_j-r_n, \ j\ne n,\ H_n^*= q-\sum\limits_{i=1}^{n-1} e_i,\ Z^*  = r_n-w-q.\nonumber
\end{equation}
Evidently, \eqref{hiresource} requires that each virus strain controls the population density of its targeted bacterial strain such that the zooplankton cannot be maintained without the presence of the resistent strain, which cannot grow too slowly.

Our main result follows. We assume that \eqref{wkBn} and \eqref{hiresource} hold.

\begin{theorem}\label{persist}
$E^*$ is a stable equilibrium and the system is permanent in the sense that
there exists $\epsilon>0$ such that every positive solution satisfies:
\begin{equation}\label{per}
H_i(t)>\epsilon, \ P_j(t)>\epsilon, \ Z(t)>\epsilon,\ t> T
\end{equation}
for all $i,j$ where $T>0$, but not $\epsilon$, depends on initial conditions.

The long term time average of each population
is its equilibrium value:
\begin{equation}\label{timeave}
\lim_{t\to\infty} \frac{1}{t}\int_0^t X(s)ds = X^*,\ X=H_i,P_j,Z,
\end{equation}
and
$H_n(t)$ and $Z(t)$ converge to their equilibrium values $H_n^*$ and $Z^*$.

Moreover,  a positive solution  either converges to $E^*$ or its  omega limit set consists of non-constant positive entire trajectories
satisfying $\sum_{i=1}^nH_i(t)=\sum_{i=1}^n H_i^*$,
 $H_n(t)= H_i^*$, $Z(t) = Z^*$,
and where $(H_i(t),P_i(t))$ is a positive solution of the classical Volterra system
  \begin{eqnarray}\label{planar}
  H_i'&=& H_i(P_i^*-P_i)\\ \nonumber
  P_i'& =& n_iP_i (H_i-H_i^*),\ 1\le i\le n-1.
\end{eqnarray}
\end{theorem}

As advertised in the introduction, Theorem~\ref{persist} says that the KtW equilibrium  is unique and stable to perturbations. More importantly, the system is permanent
in the sense that all population densities are ultimately bounded away from extinction by an initial condition independent positive quantity.
The zooplankton density and the density of the bacterial strain
resistant to virus infection converge to their equilibrium values and
if a solution does not converge to the positive equilibrium, then its long-term dynamics is described by the system consisting of $n-1$ conservative two-species systems \eqref{planar}. The latter would imply that $H_i,P_i$ are periodic with period depending on  parameters and its amplitude. However, the restriction $\sum_{i=1}^nH_i(t)=\sum_{i=1}^n H_i^*$ requires a very special resonance among the periods, suggesting that this alternative is unlikely.

Of course, our KtW model \eqref{LVmono} is  very special. Our aim was not to offer a general KtW model. Rather, it was to say as much as we could  about the long term dynamics of a KtW model and for that, we made simplifying assumptions. Most of these assumptions are also made
in the system (7.28) in \cite{W} and in similar models in the literature \cite{Fl,Jover}. It should be noted that our main result, that the KtW model is permanent, continues to hold for sufficiently small
changes in system parameters  \cite{SZ}.

Finally, we note that the main results of our earlier work \cite{KS2b}, in which we were concerned only with bacteria-virus infection networks,
can be applied to obtain results similar to Theorem~\ref{persist} for KtW models with more general infection networks than the one to one network. For example, our scaled model for the nested infection  network consisting of $n$ bacteria strains and $n$ virus strains in \cite{KS2b} is the following:
\begin{eqnarray}\label{LVnest}
H_i' &=& H_i\left(r_i-\sum\limits_{j=1}^n H_j-\sum\limits_{j \ge i}P_j\right)\\
P_i' &=& e_in_iP_i \left(\sum\limits_{j \le i}H_j-\frac{1}{e_i}\right),\ 1\le i\le n. \nonumber
\end{eqnarray}
To compare with \eqref{eqns-scaled}, set $Z=P_n$ and regard it as a zooplankton grazer. Also,  we must view the $r_i$ as $r_i-w$, $n_i=k_i+w$, and $e_i=\beta_i\phi_i/n_i$, viewed as the efficiency of virus  infection of bacteria, is comparable to the reciprocal of its value in \eqref{eqns-scaled}.
The existence of a positive equilibrium for \eqref{LVnest} requires life history trade-offs of bacteria and virus strains.
Bacteria that are more susceptible to virus infection must grow faster
\begin{equation}\label{r}
r_1>r_2>\cdots > r_n>Q_n,
\end{equation}
and the efficiency of virus infection should decline as its host range increases:
\begin{equation}\label{enorder}
e_1>e_2>e_3>\cdots>e_n.
\end{equation}
Here, $Q_n=\frac{1}{e_1}+\left(\frac{1}{e_2}-\frac{1}{e_1}\right)+\left(\frac{1}{e_3}-\frac{1}{e_2}\right)+\cdots+\left(\frac{1}{e_n}-\frac{1}{e_{n-1}}\right)$.
If \eqref{r} and \eqref{enorder} hold, there is a unique positive equilibrium $E^*$ and  all positive solutions converge to it \cite{KS2b}. By simply renaming
$Z=P_n$ and regarding it is a zooplankton, we obtain an even stronger result than Theorem~\ref{persist} for the KtW model with nested infection network provided these tradeoffs hold. Quite arbitrary infection networks among bacteria and phage might be treated using the approach in \cite{KS3}.

\begin{center}
\includegraphics[width=10cm]{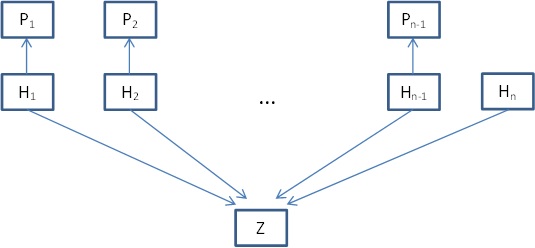}
\captionof{figure}{Interactions between the n-1 virus, n host, and the zooplankton.}
\label{fig:1}
\end{center}

\begin{center}
\includegraphics[width=12.5cm]{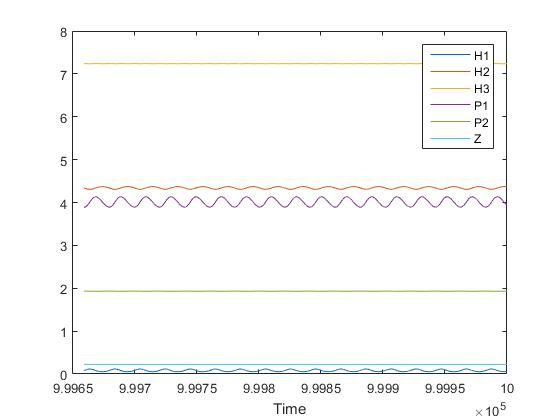}
\captionof{figure}{Last 1000 time units of a 1 million run on a population of 3 bacteria, 2 virus, and one zooplankton using ode45.  Parameters specified in the figure are chosen to satisfy conditions (\ref{wkBn}), and (\ref{hiresource}), and are not intended to be biologically realistic.  Solutions are highly oscillatory, and seem to be periodic.}
\label{fig:2}
\end{center}

\section{Proof of Main Result}
\begin{proposition}\label{bound} Solutions of \eqref{eqns-scaled} with nonnegative (positive) initial data are well-defined for all $t\ge 0$ and remain nonnegative (positive).
In addition, the system has a compact global attractor. Indeed, if
 $F(t) = \sum\limits_{i=1}^{n}H_i(t)+\sum\limits_{i=1}^{n-1}\frac{P_i(t)}{n_i}+\frac{Z}{\lambda}$ then $$F(t)\le \frac{Q}{W}+(F(0)-\frac{Q}{W})e^{-Wt} \le \max\{F(0),\frac{Q}{W}\},$$
 and
 $$\limsup_{t\to\infty}F(t)\le \sum_{i=1}^{n}(1+\frac{r_i}{W})r_i,$$
 where $K=\max\limits_{i=1}^{n}\{H_i(0),r_i\}$, $W=\min\limits_{i=1}^{n}\{e_i,w,q\}$ and $Q = \sum\limits_{i=1}^{n}(W+r_i)K$.
\end{proposition}

\begin{proof}
Existence and positivity of solutions follow from the form of the right hand side.
Therefore, $H_i'(t) \le H_i(t)(r_i-H_i(t))$.  Hence $H_i(t) \le K$ and $\limsup_{t\to\infty}H_i(t)\le r_i$.
\begin{eqnarray*}
\frac{dF}{dt} &=& \sum\limits_{i=1}^{n}(r_i-w)H_i-(\sum_{i=1}^n H_i)(\sum\limits_{j=1}^{n}H_j)-\sum\limits_{i=1}^{n}P_ie_i-Zq\\
&\le& \sum\limits_{i=1}^{n}r_iH_i-W\sum\limits_{i=1}^{n}(H_i+\frac{P_i}{e_in_i}+Z)\\
& =&\sum\limits_{i=1}^{n}(W+r_i)H_i-W F.
\end{eqnarray*}
The estimate on $F(t)$ follows by bounding the first summation by $Q$ and integrating; the estimate on the limit superior
follows from the estimate of the limit superior of the $H_i$ above and by integration.
\end{proof}

\begin{proof} {\bf Proof of Theorem~\ref{persist}.}
Since positive equilibrium $E^*$ exists, we can write \eqref{eqns-scaled} as
\begin{eqnarray}\label{LV1}
H_i' &=& H_i\left(\sum\limits\limits_{j=1}^n (H_j^*-H_j)+P_i^*-P_i+Z^*-Z\right)\\ \nonumber
H_n' &=& H_n\left(\sum\limits\limits_{j=1}^n (H_j^*-H_j)+(Z^*-Z)\right)\\ \nonumber
P_i' &=& n_iP_i (H_i-H_i^*),\ 1\le i<n\\ \nonumber
Z'   &=& Z \lambda \sum\limits\limits_{i=1}^n(H_i-H_i^*) \nonumber
\end{eqnarray}

Let $U(x,x^*)=x-x^*-x^*\log x/x^*,\ x,x^*>0$, be the familiar Volterra function and let
\begin{equation}\label{Lf}
V(H,P,Z)=\sum_{i=1}^n c_iU(H_i,H_i^*)+\sum_{j=1}^m d_jU(P_j,P_j^*)+gU(Z,Z^*)
\end{equation}
where $c_1,\cdots, c_n$ and $d_1,\cdots,d_m$ and $g$ are to be determined.

Then the derivative of $V$ along solutions of \eqref{LV1}, $\dot V$, is given by
\begin{eqnarray*}
  \dot V &=& -\left(\sum_i^n c_i(H_i-H_i^*)\right)\left(\sum_j^n (H_j-H_j^*)\right)\\
  & &-\sum_i^{n-1} c_i(H_i-H_i^*)(P_i-P_i^*) \\\nonumber
   & & -\sum_i^n c_i(H_i-H_i^*)(Z-Z^*)+\sum_i^{n-1} d_in_i(P_i-P_i^*)(H_i-H_i^*)\\
   & &+\sum_i^n g\lambda(Z-Z^*)(H_i-H_i^*)
\end{eqnarray*}
If $c_i=1$, $g=\frac{1}{\lambda}$, and $d_i=\frac{1}{n_i}$ then the last four summations cancel out and we have
\begin{equation}\label{MonoVdot}
\dot V =-\left(\sum_i H_i-\sum_i H_i^*\right)^2
\end{equation}
As $\dot V\le 0$, $E^*$ is locally stable \cite{Hale} and for each positive solution there exists $p,P>0$ such that
$p\le x(t)\le P, t\ge 0$, where $x=H_i,P_j,Z$.
Then \eqref{timeave} follows immediately from Theorem 5.2.3 in \cite{HS}.

Consider a positive solution of \eqref{LV1}. By LaSalle's invariance principle \cite{Hale,HS}, every point in its  omega limit set $L$ must satisfy $\sum_iH_i(t)=\sum_iH_i^*$ since $L\subset \{(H,V):\dot V=0\}$.
Since $V(x)\le V(H(0),P(0))$ for all $x\in L$, $L$ is a compact subset of the interior of the positive orthant.
We now consider a trajectory belonging to $L$; until further notice, all considerations involve this solution.
Since $\sum_iH_i(t)=\sum_iH_i^*$, the solution satisfies
\begin{eqnarray}\label{LV2}
H_i' &=& H_i\left(P_i^*-P_i+Z^*-Z\right)\\ \nonumber
P_i' &=& n_iP_i (H_i-H_i^*),\ 1\le i \le n-1\\ \nonumber
H_n' &=& H_n\left(Z^*-Z\right)\\ \nonumber
Z'   &=& 0 \nonumber
\end{eqnarray}
We see that $Z'\equiv 0$, therefore $Z(t)$ is a constant.  Then, $H_n'=H_n(Z^*-Z)$ so $H_n(t)$  either converges to zero, blows up to infinity, or is identically constant, depending on the value of
$Z$. The only alternative  consistent with $L$ being invariant, bounded, and bounded away from the boundary of the orthant is that $H_n(t)$ is constant
and that $Z=Z^*$. By \eqref{timeave}, it follows that $H_n=H_n^*$. Therefore \eqref{LV2} becomes:
\begin{eqnarray}\label{LV3}
H_i' &=& H_i\left(P_i^*-P_i\right)\\ \nonumber
P_i' &=& n_iP_i (H_i-H_i^*),\ 1\le i \le n-1\\ \nonumber
H_n &=& H_n^*\\ \nonumber
Z   &=& Z^* \nonumber
\end{eqnarray}
This establishes the assertions regarding \eqref{planar}. Note that as \eqref{LV3} holds on the limit set $L$ of our positive solution, it follows that $H_n(t)\to H_n^*,\ Z(t)\to Z^*$ for our positive solution.

Finally, we prove \eqref{per}.
It follows from \eqref{timeave} that $\limsup_{t\to\infty} x(t)=x^*$, for each component $x=H_i,P_j,Z$ of an arbitrary positive solution of \eqref{LV1}. This means that \eqref{LV1} is uniformly weakly persistent.  Proposition~\ref{bound} implies that the key hypotheses of Theorem 4.5 from \cite{T,ST} are satisfied, and therefore weak uniform persistence implies strong uniform persistence. This is precisely \eqref{per}.
\end{proof}

\begin{center}
\begin{tabular}{|c|c|}
  \hline
  Parameter &  value \\\hline\hline
  $r_1$ &  17.089453152634810\\\hline
  $r_2$ &  15.009830525061846\\\hline
  $r_3$ &  13.077955412892173\\\hline
  $n_1$ &  0.299362132425990\\\hline
  $n_2$ &  0.011514418415303\\\hline
  $e_1$ &  0.081255501212170\\\hline
  $e_2$ &  4.340892914457329\\\hline
  $q$ &  10.465564663600418\\\hline
  $\lambda$ &  4.474468552537804\\\hline
  $w$   & 1.194565710732100\\\hline
\end{tabular}
\captionof{figure}{Parameter values used in $Figure~\ref{fig:2}$}
\end{center}

{}
\end{document}